\documentclass[letterpaper, 10 pt, journal]{IEEEtran}
\usepackage{amsmath,amsfonts}
\usepackage{algorithmic}
\usepackage{algorithm}
\usepackage{array}

\usepackage[caption=false,font=footnotesize]{subfig}

\usepackage{textcomp}
\usepackage{stfloats}
\usepackage{url}
\usepackage{booktabs}
\usepackage{verbatim}
\usepackage{graphicx}
\usepackage{cite}
\usepackage{amssymb}
\usepackage{amsthm}
\usepackage{amsmath}
\usepackage{mathtools}

\newtheorem{theorem}{Theorem}

\newtheorem{assumption}{Assumption}

\newtheorem{remark}{Remark}

\begin{document}
\title{A Barrier-Modulated Architecture for Safe Affine Formation Control in Second-Order Multi-Agent Systems}

\author{Ashik Abrar Naeem and Mohammad Ariful Haque}

\maketitle

\begin{abstract}
Affine formation control offers immense flexibility for coordinating multi-agent maneuvers, but guaranteeing the safety of agents under parametric uncertainties remains an open challenge. This paper proposes a novel safe affine formation control framework for second-order multi-agent systems by integrating Higher-Order Control Barrier Functions (HOCBFs) with Adaptive Dynamic Programming (ADP). We introduce a barrier-modulated control architecture that smoothly attenuates the nominal formation tracking objective when agents approach safety boundaries, preventing conflicting control inputs. Within this architecture, two distinct safety controllers are developed: (1) an analytical barrier-gradient repulsive controller that provides a computationally efficient, rigorous mathematical baseline, and (2) a data-driven optimal safety controller. The data-driven approach utilizes an actor-critic neural network to solve the Hamilton-Jacobi-Bellman (HJB) equation online, enabling optimal collision avoidance even in the presence of unknown system parameters. Using Nagumo's theorem and Lyapunov stability analysis, we formally prove that both controllers guarantee the forward invariance of the safe set ensuring absolute collision avoidance while maintaining Uniformly Ultimately Bounded (UUB) formation tracking errors. Finally, simulations validate the theoretical findings and demonstrate the robustness of the proposed controllers in dynamic obstacle avoidance scenarios.
\end{abstract}


\section{Introduction}
\IEEEPARstart{M}{ulti-agent} formation control has been studied extensively since 1990, when Sugihara and Suzuki \cite{sugihara1990distributed} proposed a simple algorithm for a group of point-mass robots to approximate circles and simple polygons. As the discipline matured, it evolved to encompass a variety of complex agent dynamics and communication topologies. The necessary and sufficient graphical conditions for formation stabilization was established in \cite{lin2005necessary}, while consensus algorithms \cite{ren2008consensus} and distributed containment control strategies \cite{cao2010distributed} significantly improved the coordination of double integrator dynamics. Complex motion constraints were integrated \cite{zhao2017general} into these theoretical models to enhance practical deployment. More Recently, deep reinforcement learning has been applied to adaptive cooperative trajectory planning \cite{xing2024multi}. A critical aspect of multi-agent coordination is how the target geometric pattern is defined and manipulated. While many traditional approaches rely on relative positions or distances, breakthroughs in bearing rigidity have enabled almost global bearing-only formation stabilization \cite{zhao2015bearing}. Building on this, bearing-based approaches have been successfully utilized for translational and scaling formation maneuver control \cite{zhao2015translational}. However, affine formation control has emerged as a powerful framework to achieve even more flexible and complex geometric maneuvers, such as shearing, rotation, and scaling.

Affine transformation is a general linear transformation that may correspond to a translation, rotation, scaling, shear, or compositions of them. Affine transformation preserves straight lines and planes. As a result, collinear (or coplanar) points remain collinear (or coplanar) after any affine transformations. Parallel lines are also preserved by affine transformations. Affine formation control laws preserve stability when the nominal configuration undergoes arbitrary affine transformations. Lin et al. \cite{lin2015necessary} first proposed affine formation control as a distributed multi-agent framework and established necessary and sufficient graphical conditions for both undirected and directed interaction graphs. Later, Zhao \cite{zhao2018affine} introduced the notion of affine localizability and proposed distributed affine formation maneuver control laws for single and double-integrator agent dynamics along with non-holonomic motion and velocity saturation constraints. \cite{zhi2021leader} further progressed the control laws for second-order nonlinear uncertain systems. The problem of constructing affine frameworks in three scenarios e.g. vertex addition, edge deletion, and vertex deletion is analyzed in \cite{li2025distributed}. These works on affine formation frameworks neglect the safety of the maneuvering agents. Recently, \cite{liu2025flexible} utilized deep reinforcement learning for collision avoidance. But in their approach, the followers need to know the global position of the leaders which violates the distributed affine formation framework. Further analysis suggests that reinforcement learning is inherently uncertain. So, formal safety ensurance in affine formation control is still an open problem.

Intuitively, safety requires that “bad” things do not happen. In the context of dynamical systems, safety is studied since 1940's, when Nagumo provided necessary and sufficient conditions for set invariance \cite{nagumo1942lage}. But a more formal study of safety started from 2014 in the context of control barrier functions (CBFs) \cite{ames2014control}. Since then several forms of control barrier functions have proposed \cite{ames2014control, romdlony2016stabilization, ames2019control, nguyen2016exponential, xiao2021high, breeden2023robust, li2023survey}. While quadratic programming (QP) is a common method of utilizing CBFs \cite{emam2022safe, zhang2025output}, standard formulations can induce undesired, non-origin equilibrium points in the closed-loop system \cite{tan2024undesired}. Recently, adaptive dynamic programming (ADP) has gained significant attention due to the ability of neural networks to approximate the HJB and HJI equations.

Reinforcement learning (RL) has been an effective tool in solving nonlinear optimization problems \cite{vrabie2014optimal}. A well-known architecture used in RL is the actor-critic structure. Adaptive dynamic programming \cite{liu2017adaptive} also takes the actor-critic structure as an implementation architecture and shares similar spirits as RL. Thus researchers often use ADP and RL as two interchangeable terms. Critic-only architectures \cite{yang2020optimal, qiao2022asymmetric, yang2018adaptive} have become increasingly popular due to their ability to solve the HJB and HJI equations using a single neural network. However, these methods inherently require complete knowledge of the system dynamics. In scenarios characterized by uncertainty, actor-critic structures \cite{kamalapurkar2015approximate, abu2005nearly, vamvoudakis2010online} offer a highly effective alternative for robust approximation. The synthesis of CBFs and ADP constitutes a nascent area of study that is presently attracting significant research interest. Marvi et. al. \cite{marvi2021safe} presented a actor-critic architecture-based barrier certified method to learn safe optimal controllers and established formal proofs. While \cite{zhang2024safe} designed a safe and optimal robust control method to ensure the system operated in a safe region under the influence of external disturbances. Wang et. al. \cite{wang2026reinforcement} treated external disturbances adversarial agent and transformed the safe optimal formation tracking control into a constrained zero-sum (ZS) differential game and designed and integrated a CBF into the cost function to penalize the unsafe behavior.

This study integrates control barrier functions with slide mode controller to design safe affine formation control for second-order multi-agent systems. Specifically, by separating the nominal formation task from the collision avoidance mechanism, this work enable agents to dynamically prioritize safety in uncertain environments without permanently sacrificing the target geometric pattern. 

The main contributions of this paper are summarized as follows:
\begin{itemize}
    \item \textbf{Barrier-Modulated Control Architecture:} We propose a novel framework that integrates a nominal affine formation tracking controller with safety mechanisms. This architecture utilizes a smooth barrier-modulation function that attenuates the nominal controller only near safety boundaries to prevent collision among the agent and ensuring maneuverability.
    \item \textbf{Analytical Barrier-Gradient Safety Controller:} We develop a computationally efficient safety controller based on spatial barrier-gradients. This approach provides a rigorous analytical baseline that formally guarantees the forward invariance of the safe set and uniformly ultimately bounded tracking errors, offering a highly reliable collision avoidance mechanism.
    \item \textbf{Optimal Safety Control via ADP:} We design a data-driven safety controller using an actor-critic neural network structure to solve the Hamilton-Jacobi-Bellman (HJB) equation online. By defining a danger-weighted state representation, the network learns an optimal safety policy that dynamically adapts to parametric uncertainties without requiring exact knowledge of the system dynamics.
    \item \textbf{Rigorous Safety and Stability Guarantees:} We provide formal proofs establishing the forward invariance of the safe set via Nagumo's theorem. Furthermore, rigorous Lyapunov stability analysis is conducted to prove the Uniform Ultimate Boundedness (UUB) of the tracking errors and neural network weight estimation errors.
\end{itemize}

The remainder of this paper is organized as follows: Section II introduces the theoretical preliminaries of affine formation and formulates the control problem. Section III outlines the barrier-modulated control architecture. Section IV presents the analytical barrier-gradient safety controller, while Section V details the ADP-based safety controller. Simulation results validating both approaches are discussed in Section VI. Finally, Section VII concludes the paper.

\section{Preliminaries and Problem Formulation}

This section is composed of two parts. First the concept of affine formation is introduced, and then the problem is formulated.

\subsection{Affine Formation}

The formation for a group of $n$ agents in $\mathbb{R}^d$ is denoted as $(\mathcal{G}, p)$, where the graph $\mathcal{G} = (\mathcal{V}, \mathcal{E})$ consists of a vertex set $\mathcal{V} = \{1, \dots, n\}$ and an edge set $\mathcal{E} \subseteq \mathcal{V} \times \mathcal{V}$. The edge $(i,j) \in \mathcal{E}$ indicates that agent $i$ can receive information from agent $j$, and agent $j$ is a neighbor of $i$. The set of neighbors of vertex $i$ is $\mathcal{N}_i = \{j \in \mathcal{V} : (i,j) \in \mathcal{E}\}$. This paper only considers undirected graphs, where $(i,j) \in \mathcal{E} \Leftrightarrow (j,i) \in \mathcal{E}$. Let $p_i \in \mathbb{R}^d$ be the position of agent $i$ and $p = [p_1^T, \dots, p_n^T]^T \in \mathbb{R}^{dn}$ be the configuration of all the agents.  Suppose the first $n_\ell$ agents are leaders and the rest $n_f = n - n_\ell$ agents are followers. Let $\mathcal{V}_\ell = \{1, \dots, n_\ell\}$ and $\mathcal{V}_f = \mathcal{V} \setminus \mathcal{V}_\ell$ be the sets of leaders and followers, respectively. The positions of the leaders and followers are denoted as $p_\ell = [p_1^T, \dots, p_{n_\ell}^T]^T$ and $p_f = [p_{n_\ell+1}^T, \dots, p_n^T]^T$, respectively. The constant configuration $r$ represents a typical geometric pattern that the formation would like to maintain. Here, $r$ is called the nominal configuration and $(\mathcal{G},r)$ the nominal formation. The notion along with the necessary and sufficient conditions of \textit{affine formation control} was first proposed in \cite{lin2015necessary}. The \textit{affine image}, the set consisting of all the affine transformations of the nominal configuration $r$, defined as
\begin{equation*}
    \begin{split}
        \mathcal{A}(r) &= \{p \in \mathbb{R}^{dn} : p = (I_n \otimes A)r + \mathbf{1}_n \otimes b, \\
        &\qquad A \in \mathbb{R}^{d \times d}, \; b \in \mathbb{R}^d \}.
    \end{split}
\end{equation*}

For formation $(\mathcal{G}, p)$, a stress is a set of scalars, $\{\omega_{ij}\}_{(i,j) \in \mathcal{E}}$, where $\omega_{ij} = \omega_{ji} \in \mathbb{R}$, assigned to all the edges. A stress is called an equilibrium stress \cite{connelly2005generic, alfakih2008dual, connelly2015frameworks}, if it satisfies
\begin{equation}
\sum_{j \in \mathcal{N}_i} \omega_{ij}(p_j - p_i) = 0, \quad i \in \mathcal{V}. \label{eq:equilibrium_stress}
\end{equation}
Equation \eqref{eq:equilibrium_stress} can be expressed in a matrix form as
$$(\Omega \otimes I_d)p = 0$$
where $\Omega \in \mathbb{R}^{n \times n}$ is the stress matrix satisfying
$$
[\Omega]_{ij} = 
\begin{cases}
0, & i \neq j, (i,j) \notin \mathcal{E} \\
-\omega_{ij}, & i \neq j, (i,j) \in \mathcal{E} \\
\sum_{k \in \mathcal{N}_i} \omega_{ik}, & i = j.
\end{cases}
$$ Every configuration in the affine image satisfies the equilibrium stress condition \eqref{eq:equilibrium_stress} if the nominal configuration $\{r_i\}_{i=1}^n$ affinely spans $\mathbb{R}^d$. The necessary and sufficient conditions to establish this statement were first analyzed in terms of linear span in \cite{lin2015necessary}, and affine span in \cite{zhao2018affine}. The notion of \textit{affine localizability} along with its necessary and sufficient conditions established by \cite{zhao2018affine}. The nominal formation $(\mathcal{G}, r)$ is affinely localizable by the leaders if for any $p = [p_\ell^T, p_f^T]^T \in \mathcal{A}(r)$, $p_f$ can be uniquely determined by $p_\ell$.

\begin{assumption}
We assume that
\begin{enumerate}
    \item The nominal configuration $\{r_i\}_{i=1}^n$ affinely spans $\mathbb{R}^d$.
    \item The nominal formation $(\mathcal{G}, r)$ possesses a positive semi-definite stress matrix $\Omega$ satisfying $\mathrm{rank}(\Omega) = n - d - 1$.
    \item The nominal formation $(\mathcal{G}, r)$ is affinely localizable by the leaders.
    \item The leaders' trajectories are governed by an external mechanism ensuring $p_l(t) = p^*_l(t)$ for all $t$.
\end{enumerate}
\label{as:preliminary}
\end{assumption}

\subsection{Problem Formulation}

Assume that all followers are governed by the following nonlinear dynamics
\begin{equation}
\begin{aligned}
    \dot{p}_i &= v_i \\
    \dot{v}_i &= u_i + \Phi_i \theta_i, \quad i = 1, \dots, n_f
\end{aligned}
\label{eq:dynamics}
\end{equation}
where $p_i \in \mathbb{R}^d$ denotes agent $i$'s position, $v_i \in \mathbb{R}^d$ denotes agent $i$'s velocity, $u_i \in \mathbb{R}^d$ is the input of agent $i$, and $\Phi_i(p_i,v_i) \in \mathbb{R}^{d \times q}$ is a matrix whose elements are bounded and continuous functions of $p_i$ and $v_i$. The term $\theta_i \in \mathbb{R}^q$ is the unknown and constant parameter vector of agent $i$, e.g., the mass or the fluid density in the environment.

The objective is to steer the followers to track a time-varying target formation maintaining a minimum safe distance from other agents. The time-varying configuration of the target formation has the form of
$$p^*(t) = [I_n \otimes A(t)]r + \mathbf{1}_n \otimes b(t)$$
where $r = [r_1^T, \dots, r_n^T]^T = [r_\ell^T, r_f^T]^T \in \mathbb{R}^{dn}$ is the nominal configuration, and $A(t) \in \mathbb{R}^{d \times d}$ and $b(t) \in \mathbb{R}^d$ are continuous functions of $t$. The target configuration is a time-varying affine transformation of the nominal configuration.

\section{Barrier-Modulated Control Architecture}
To achieve simultaneous formation tracking and collision avoidance a barrier-modulated control architecture is proposed. The control input for follower $i$ is
\begin{equation}
    u_i = \rho_i \, u_{\mathrm{nom},i} + u_{\mathrm{safe},i}
    \label{eq:control_input}
\end{equation}
where $u_{\mathrm{nom},i}$ is the nominal formation tracking controller, $u_{\mathrm{safe},i}$ is the safety controller, and $\rho_i \in [0,1]$ is a smooth barrier-modulation function that attenuates the nominal controller near the safety boundary. Two designs for $u_{\mathrm{safe},i}$ are proposed: an analytical barrier-gradient controller (Section~\ref{sec:barrier}) and an adaptive dynamic programming (ADP) controller (Section~\ref{sec:adp}).

\subsection{Sensor Model}

Each follower $i$ communicates with its graph neighbors $\mathcal{N}_i$ for formation control. Additionally, an onboard sensor detects all agents within a local sensing radius $R_{\rm sense}$. Let
\begin{equation}
    \mathcal{S}_i(t) = \{j \in \mathcal{V} \setminus \{i\} : \|p_i(t) - p_j(t)\| < R_{\rm sense}\}
    \label{eq:sensor_set}
\end{equation}
be the set of agents detected by the sensor of follower $i$ at time $t$. All safety-related quantities such as the modulation function $\rho_i$, the barrier penalty $B_i$, and the safety controller $u_{\mathrm{safe},i}$ are computed over $\mathcal{S}_i(t)$. This ensures that potential collisions with non-neighbor agents are also detected and avoided.

\subsection{Safety Functions}

Let $D_s$ denote the minimum allowable safe distance. The zeroth-order safety constraint is
\begin{equation}
    h_0(x_i, x_j) = \|p_i - p_j\|^2 - D_s^2 \geq 0.
    \label{eq:h0_constraint}
\end{equation}

Because $u_i$ acts on the acceleration, a Higher-Order Control Barrier Function (HOCBF) is employed. The modified safety function is
\begin{equation}
    h_{\mathrm{safe}}(x_i, x_j) = 2(p_i - p_j)^T(v_i - v_j) + \gamma \left( \|p_i - p_j\|^2 - D_s^2 \right),
    \label{eq:h_safe_final}
\end{equation}
where $\gamma > 0$ is a design parameter. The barrier-modulation function is
\begin{equation}
    \rho_i = 1 - \exp\!\Big(-\beta \max\big(0,\; \min_{j \in \mathcal{S}_i} h_0(x_i, x_j)\big)\Big),
    \label{eq:rho_def}
\end{equation}
where $\beta > 0$ controls the steepness. When all sensed agents are far ($h_0 \gg 0$), $\rho_i \approx 1$ and the nominal controller operates at full capacity. As any agent approaches the safety boundary ($h_0 \to 0$), $\rho_i \to 0$, attenuating the nominal controller. The barrier penalty aggregating all pairwise safety costs is
\begin{equation}
    B_i(x_i) = \sum_{j \in \mathcal{S}_i} \frac{\mu}{\max(h_{\mathrm{safe}},\, \epsilon_b)},
    \label{eq:barrier_penalty}
\end{equation}
where $\mu > 0$ is a weighting parameter and $\epsilon_b > 0$ prevents division by zero.

\subsection{Nominal Controller}
The nominal controller is borrowed from \cite{zhi2021leader} and formulated by the following algorithm

\begin{align}
    u_i &= -s_i - \hat{\gamma}_i \circ \operatorname{sgn}(s_i) - \Phi_i \hat{\theta}_i \label{eq:control_u} \\
    \dot{\hat{\gamma}}_i &= c_1 \operatorname{abs}(s_i) \label{eq:gamma_dot} \\
    \dot{\hat{\theta}}_i &= c_2 \Phi_i^T s_i \label{eq:theta_dot}
\end{align}
where $s_i = \sum_{j=1}^{m+n} \omega_{ij} (p_i - p_j) + a \sum_{j=1}^{m+n} \omega_{ij} (v_i - v_j)$, $\hat{\theta}_i \in \mathbb{R}^q$, $\hat{\gamma}_i \in \mathbb{R}^d$ and $a, c_1, c_2 \in \mathbb{R}$ are all positive constant gains.

\begin{theorem}[\cite{zhi2021leader}]
Considering the system \eqref{eq:dynamics} suppose that Assumption \ref{as:preliminary} holds, the acceleration of the leaders is bounded, and the velocity of the leaders is continuous. If
\begin{equation}
    \lambda_{\min}(\bar{\Omega}_{ff}) > a^{-3}
    \label{eq:lambda_condition}
\end{equation}
is satisfied, then the control law \eqref{eq:control_u}-\eqref{eq:theta_dot} can guarantee that the followers will converge to their desired positions.
\label{th:nominal_controller}
\end{theorem}

\section{Safety Controller I: Barrier-Gradient Repulsion}
\label{sec:barrier}

The first safety controller is an analytical barrier-gradient repulsive force computed over the sensor set
\begin{equation}
    u_{\mathrm{safe},i} = \sum_{j \in \mathcal{S}_i} \frac{2\kappa\mu(p_i - p_j)}{\max(h_{\mathrm{safe}},\, 0)^2 + \epsilon^2},
    \label{eq:u_repel}
\end{equation}
where $\kappa > 0$ is the repulsion gain and $\epsilon > 0$ is a regularization parameter. This force pushes agent $i$ away from each sensed agent $j$, with magnitude that grows as $h_{\mathrm{safe}} \to 0$. The step-by-step execution of this barrier-gradient safety mechanism is detailed in Algorithm~\ref{alg:barrier_gradient}. With the control architecture fully defined, we now establish the rigorous safety guarantees and closed-loop stability properties of the system.

\begin{algorithm}[htbp]
	\caption{Barrier-Gradient Repulsion Controller for Follower $i$}
	\label{alg:barrier_gradient}
	\begin{algorithmic}[1] 
		\REQUIRE Current state $x_i(t) = (p_i, v_i)$, nominal control $u_{\mathrm{nom},i}$, safety distance $D_s$, sensor radius $R_{\rm sense}$, and tuning parameters $\gamma, \beta, \kappa, \mu, \epsilon > 0$.
		\ENSURE Total safety-modulated control input $u_i(t)$.
		\STATE Obtain local sensor measurements to form the set of detected agents:
		\quad $\mathcal{S}_i(t) = \{j \in \mathcal{V} \setminus \{i\} : \|p_i - p_j\| < R_{\rm sense}\}$
		\STATE Initialize safety control vector: $u_{\mathrm{safe},i} \leftarrow \mathbf{0} \in \mathbb{R}^d$
		\IF{$\mathcal{S}_i(t)$ is empty}
		\STATE No collision threat detected; full nominal control authorized.
		\STATE $\rho_i \leftarrow 1$
		\ELSE
		\STATE Initialize minimum distance tracker: $h_{0,\min} \leftarrow \infty$
		\FOR{\textbf{each} detected agent $j \in \mathcal{S}_i(t)$}
		\STATE Compute spatial safety constraint:
		\STATE \quad $h_0(x_i, x_j) \leftarrow \|p_i - p_j\|^2 - D_s^2$
		\STATE Update closest threat boundary:
		\STATE \quad $h_{0,\min} \leftarrow \min\big(h_{0,\min}, \; h_0(x_i, x_j)\big)$
		\STATE Compute the Higher-Order Control Barrier Function (HOCBF):
		\STATE \quad $h_{\mathrm{safe}}(x_i, x_j) \leftarrow 2(p_i - p_j)^T(v_i - v_j) + \gamma h_0(x_i, x_j)$
		\STATE Accumulate barrier-gradient repulsive force:
		\STATE \quad $u_{\mathrm{safe},i} \leftarrow u_{\mathrm{safe},i} + \frac{2\kappa\mu(p_i - p_j)}{\max(h_{\mathrm{safe}}(x_i, x_j),\, 0)^2 + \epsilon^2}$
		\ENDFOR
		\STATE Compute the barrier-modulation function based on the closest threat:
		\STATE \quad $\rho_i \leftarrow 1 - \exp\Big(-\beta \max\big(0, h_{0,\min}\big)\Big)$
		\ENDIF
		\STATE Apply the barrier-modulated architecture:
		\STATE \quad $u_i(t) \leftarrow \rho_i \, u_{\mathrm{nom},i} + u_{\mathrm{safe},i}$
		\RETURN $u_i(t)$
	\end{algorithmic}
\end{algorithm}

\subsection{Safety and Stability Analysis}

The following analysis establishes the collision avoidance and closed-loop stability properties of the proposed barrier-modulated architecture. To facilitate the stability proofs, we formally bound the system parameters and control components.

\begin{assumption}[System Boundedness]
	\label{as:bounded_dynamics}
	We assume that for all agents $i \in \mathcal{V}$ and all $t \ge 0$ the nominal formation controller, the parametric uncertainties and the sum of the safety repulsive forces from all non-critical agents ($k \in \mathcal{S}_i \setminus \{j\}$) acting on follower $i$ is bounded such that $\|u_{\mathrm{nom},i}\| \le \bar{U}_{\mathrm{nom}}$, $\|\Phi_i(p_i, v_i)\theta_i\| \le \bar{\Theta}$ and $\|\sum_{k \neq j} u_{\mathrm{safe},ik}\| \le \bar{U}_{\mathrm{other}}$.  
\end{assumption}

\begin{theorem}[Barrier-Gradient Controller Safety]
	\label{thm:safety_barrier}
	Consider the multi-agent system described by \eqref{eq:dynamics} with the control law \eqref{eq:control_input} and \eqref{eq:u_repel}. Suppose Assumptions \ref{as:preliminary} and \ref{as:bounded_dynamics} hold, and the initial states reside strictly within the safe set, i.e., $h_0(x_i(0), x_j(0)) > 0$ and $h_{\mathrm{safe}}(x_i(0), x_j(0)) > 0$ for all pairs $(i,j)$. If the repulsion gain $\kappa$, penalty weight $\mu$, and regularization parameter $\epsilon$ are chosen to satisfy the sufficient gain condition
	\begin{equation*}
		\frac{8\kappa\mu D_s^2}{\epsilon^2} > 4R_{\rm sense}\Big(\bar{U}_{\mathrm{nom}} + \bar{\Theta} + \bar{U}_{\mathrm{other}}\Big) + \gamma^2(R_{\rm sense}^2 - D_s^2)
		\label{eq:sufficient_gain_explicit}
	\end{equation*}
	then the safe set $\mathcal{C} = \{x \mid h_0(x_i, x_j) > 0, \; \forall j \in \mathcal{S}_i\}$ is forward invariant, guaranteeing collision avoidance for all $t \ge 0$.
\end{theorem}

\begin{proof}
	By Nagumo's Theorem \cite{nagumo1942lage}, $\mathcal{C}$ is forward invariant if $\dot{h}_{\mathrm{safe}}|_{h_{\mathrm{safe}} \to 0^+} > 0$. Let $p_{ij} \triangleq p_i - p_j$, $v_{ij} \triangleq v_i - v_j$. Applying High-Order Control Barrier Function (HOCBF) boundary conditions \cite{xiao2021high}, $2p_{ij}^T v_{ij} \to -\gamma h_0(x_i, x_j) \implies$
	\begin{equation*}
		\dot{h}_{\mathrm{safe}}\Big|_{h_{\mathrm{safe}} \to 0^+} = 2\|v_{ij}\|^2 + 2p_{ij}^T \dot{v}_{ij} - \gamma^2 h_0(x_i, x_j).
	\end{equation*}
	
	Substituting $\dot{v}_i = \rho_i u_{\mathrm{nom},i} + u_{\mathrm{safe},i} + \Phi_i \theta_i$ and isolating the principal repulsive force yields
	\begin{equation*}
		2p_{ij}^T \dot{v}_{ij} = \frac{8\kappa\mu \|p_{ij}\|^2}{\epsilon^2} + 2p_{ij}^T w,
	\end{equation*}
	where $w \triangleq \rho_i u_{\mathrm{nom},i} - \rho_j u_{\mathrm{nom},j} + \Phi_i \theta_i - \Phi_j \theta_j + U_{i,\mathrm{other}} - U_{j,\mathrm{other}}$.
	
	Applying the Cauchy-Schwarz and Triangle inequalities to bound the worst-case disturbances, with $\|p_{ij}\| \le R_{\rm sense}$ and Assumption \ref{as:bounded_dynamics}
	\begin{equation*}
		2p_{ij}^T w \ge -4R_{\rm sense} \Big( \bar{U}_{\mathrm{nom}} + \bar{\Theta} + \bar{U}_{\mathrm{other}} \Big).
	\end{equation*}
	
	Noting $2\|v_{ij}\|^2 \ge 0$, $\|p_{ij}\| \ge D_s$, and $h_0(x_i, x_j) \le R_{\rm sense}^2 - D_s^2 \implies -\gamma^2 h_0(x_i, x_j) \ge -\gamma^2(R_{\rm sense}^2 - D_s^2)$, the boundary derivative satisfies
	\begin{equation*}
		\begin{split}
			\dot{h}_{\mathrm{safe}}\Big|_{h_{\mathrm{safe}} \to 0^+} &\ge \frac{8\kappa\mu D_s^2}{\epsilon^2} - 4R_{\rm sense}\Big(\bar{U}_{\mathrm{nom}} + \bar{\Theta} + \bar{U}_{\mathrm{other}}\Big) \\
			&\quad - \gamma^2(R_{\rm sense}^2 - D_s^2).
		\end{split}
	\end{equation*}
	
	By \eqref{eq:sufficient_gain_explicit}, $\dot{h}_{\mathrm{safe}}|_{h_{\mathrm{safe}} \to 0^+} > 0 \implies \mathcal{C}$ is forward invariant.
\end{proof}

\begin{theorem}[Barrier-Gradient Controller Stability]
	\label{thm:uub_barrier}
	Consider the multi-agent system described by \eqref{eq:dynamics} with the control law \eqref{eq:control_input} and \eqref{eq:u_repel}. Suppose Assumptions \ref{as:preliminary} and \ref{as:bounded_dynamics} hold, and the safety conditions of Theorem \ref{thm:safety_barrier} are satisfied, ensuring the safe set $\mathcal{C}$ is forward invariant. Then, the formation tracking error, along with the adaptive estimation errors, are Uniformly Ultimately Bounded (UUB).
\end{theorem}

\begin{proof}
	Let $e_{f} \in \mathbb{R}^{d n_f}$, $s_f = \bar{\Omega}_{ff} (e_f + a \dot{e}_f) \in \mathbb{R}^{d n_f} \implies \dot{e}_f = \frac{1}{a}(\bar{\Omega}_{ff}^{-1} s_f - e_f)$. Define $\tilde{\theta}_i = \theta_i - \hat{\theta}_i$ and $\tilde{\gamma}_i = \gamma_i - \hat{\gamma}_i$. Given $\dot{\theta}_i = 0, \dot{\gamma}_i = 0 \implies \dot{\tilde{\theta}}_i = -\dot{\hat{\theta}}_i, \dot{\tilde{\gamma}}_i = -\dot{\hat{\gamma}}_i$.
	
	Consider the Lyapunov candidate
	\begin{equation*}
    \begin{split}
        V &= \frac{1}{2a} s_f^T \bar{\Omega}_{ff}^{-1} s_f + \frac{1}{2a^2} e_f^T \bar{\Omega}_{ff}^{-1} e_f \\
        &\quad + \sum_{i \in \mathcal{V}_f} \left( \frac{1}{2 c_2} \tilde{\theta}_i^T \tilde{\theta}_i + \frac{1}{2 c_1} \tilde{\gamma}_i^T \tilde{\gamma}_i \right).
    \end{split}
    \end{equation*}
	
	Let $u_i = u_{\mathrm{nom},i} + \varepsilon_i(t)$, where $\varepsilon_i(t) = (\rho_i - 1)u_{\mathrm{nom},i} + u_{\mathrm{safe},i}$. By Theorem \ref{thm:safety_barrier}, $h_{\mathrm{safe}} > 0 \implies \exists \bar{\varepsilon} > 0 \text{ s.t. } \|\varepsilon(t)\| \le \bar{\varepsilon}, \; \forall t \ge 0$.
	
	Differentiating $V$, substituting $\dot{\hat{\theta}}_i = c_2 \Phi_i^T s_i, \dot{\hat{\gamma}}_i = c_1 \operatorname{abs}(s_i)$ \cite{zhi2021leader}, and applying $-\gamma_i^T \operatorname{abs}(s_i) - s_i^T \ddot{p}_{f,i}^* \le 0$ \cite{slotine1991applied} yields:
	\begin{equation*}
		\dot{V} \le -s_f^T \left( I - \frac{1}{a^2} \bar{\Omega}_{ff}^{-1} \right) s_f - \frac{1}{a^3} e_f^T \bar{\Omega}_{ff}^{-1} e_f + s_f^T M e_f + s_f^T \varepsilon(t),
	\end{equation*}
	where $M \triangleq \frac{1}{a^3} \bar{\Omega}_{ff}^{-2} - \frac{1}{a^2} I$.
	
	By Young's inequality, $\exists \eta > 0$ such that $s_f^T M e_f \le \frac{\eta}{2} \|s_f\|^2 + \frac{\|M\|^2}{2\eta} \|e_f\|^2$. For a sufficiently large nominal gain $a > 0$ \cite{khalil2002nonlinear}
	\begin{equation*}
    \begin{split}
        \alpha &\triangleq \min \Bigg( \lambda_{\min}\Big(I - \frac{1}{a^2} \bar{\Omega}_{ff}^{-1}\Big) - \frac{\eta}{2}, \\
        &\qquad\qquad \frac{\lambda_{\min}(\bar{\Omega}_{ff}^{-1})}{a^3} - \frac{\|M\|^2}{2\eta} \Bigg) > 0.
    \end{split}
\end{equation*}
	
	Defining $Z_{\mathrm{track}} = [e_f^T, s_f^T]^T \in \mathbb{R}^{2dn_f}$ and constant $\vartheta \in (0, 1)$
	\begin{equation*}
	\begin{split}
		\dot{V} &\le -\alpha \|Z_{\mathrm{track}}\|^2 + \|Z_{\mathrm{track}}\| \bar{\varepsilon} \\
		&\le -(1 - \vartheta)\alpha \|Z_{\mathrm{track}}\|^2 - \vartheta \alpha \|Z_{\mathrm{track}}\|^2 + \|Z_{\mathrm{track}}\| \bar{\varepsilon}.
	\end{split}
    \end{equation*}
	\begin{equation*}
		\dot{V} < 0 \quad \forall \|Z_{\mathrm{track}}\| > \frac{\bar{\varepsilon}}{\vartheta \alpha}.
	\end{equation*}
	
	$\implies Z_{\mathrm{track}}$ is Uniformly Ultimately Bounded (UUB) \cite{khalil2002nonlinear}.
\end{proof}

\section{Safety Controller II: Adaptive Dynamic Programming}
\label{sec:adp}

The second safety controller uses an actor-critic neural network to learn the optimal safety policy online via Adaptive Dynamic Programming (ADP). When the safety-modulated control architecture is active, the nominal control effort is attenuated by the barrier function $\rho_i \in [0, 1]$. To maintain closed-loop stability and ensure proper Lyapunov cancellations, the adaptive estimation laws must be proportionally modified
\begin{align}
	\dot{\hat{\gamma}}_i &= \rho_i \, c_1 |s_i|, \label{eq:gamma_dot_adp} \\
	\dot{\hat{\theta}}_i &= \rho_i \, c_2 \Phi_i^T s_i. \label{eq:theta_dot_adp}
\end{align}

\begin{remark}
	The $\rho_i$ factor in \eqref{eq:gamma_dot_adp}--\eqref{eq:theta_dot_adp} serves a critical practical function: it acts as an adaptive gain freezing mechanism. When an agent approaches a safety boundary, the nominal controller is severely attenuated ($\rho_i \ll 1$). Consequently, the sliding variable $s_i$ does not converge to zero because the safety override intentionally pushes the agent off its nominal tracking trajectory. While the analytical barrier-gradient controller can absorb the attenuation of the nominal tracking objective as a bounded perturbation without modifying the adaptive laws, the Adaptive Dynamic Programming (ADP) architecture strictly requires the adaptive gain freezing mechanism.
	\label{rem:freezing}
\end{remark}

To ensure the neural networks receive a meaningful signal when safety is critical, a danger-weighted input state $z_i = [z_{p,i}^T, z_{v,i}^T]^T \in \mathbb{R}^{2d}$ is defined
\begin{equation}
	z_i = \sum_{j \in \mathcal{S}_i} w_{ij} \begin{bmatrix} p_i - p_j \\ v_i - v_j \end{bmatrix}, \quad w_{ij} = \frac{D_s^4}{\left(\|p_i - p_j\|^2 + D_s^2\right)^2},
	\label{eq:danger_state}
\end{equation}
where $w_{ij} \in (0,1]$ assigns larger weight to closer agents. Define the total danger weight $w_{\Sigma,i} = \sum_{j \in \mathcal{S}_i} w_{ij}$. Taking the partial derivative of the velocity component $z_{v,i} \in \mathbb{R}^d$ with respect to the agent's physical velocity $v_i$ yields
\begin{equation}
	\frac{\partial z_{v,i}}{\partial v_i} = w_{\Sigma,i} \, I_d.
	\label{eq:chain_rule}
\end{equation}

\subsection{Optimal Safety Control via HJB}

To synthesize the optimal safety policy, we formulate an infinite-horizon discounted optimal control problem. The cost functional for agent $i$ is defined as:
\begin{equation}
    \begin{split}
        V_i(z_i(t)) &= \int_t^\infty e^{-\alpha(\tau-t)} \Big( B_i(z_i(\tau)) \\
        &\quad + u_{\mathrm{safe},i}^T(\tau) R_i \, u_{\mathrm{safe},i}(\tau) \Big) d\tau,
    \end{split}
    \label{eq:adp_cost}
\end{equation}
where $\alpha > 0$ is the discount factor ensuring the finiteness of the integral, and $R_i \in \mathbb{R}^{d \times d}$ is a symmetric positive-definite weighting matrix penalizing excessive safety interventions. Let $V_i^*(z_i)$ denote the optimal value function that minimizes \eqref{eq:adp_cost}.

To solve for the optimal policy using the Hamilton-Jacobi-Bellman (HJB) framework, we must isolate the safety control input $u_{\mathrm{safe},i}$ within the time derivative of the danger state $\dot{z}_i$. By separating the drift dynamics from the control input, the evolution of $z_i$ can be expressed in a control-affine form
\begin{equation}
	\dot{z}_i = F_i(x) + G_i(x) u_{\mathrm{safe},i}
\end{equation}
where $F_i(x)$ encapsulates the drift driven by the nominal controller, neighbors' states, and external disturbances. The input coupling matrix $G_i(x) \in \mathbb{R}^{2d \times d}$ dictates how the safety control specifically influences the danger state. Using the chain rule established in \eqref{eq:chain_rule}, we explicitly construct $G_i(x)$ as
\begin{equation}
	G_i(x) = \frac{\partial z_i}{\partial v_i} = \begin{bmatrix} \mathbf{0}_{d \times d} \\ \frac{\partial z_{v,i}}{\partial v_i} \end{bmatrix} = \begin{bmatrix} \mathbf{0}_{d \times d} \\ w_{\Sigma,i} I_d \end{bmatrix}
	\label{eq:input_coupling}
\end{equation}

The optimal value function $V_i^*(z_i)$ must satisfy the corresponding continuous-time HJB equation
\begin{equation}
	\min_{u_{\mathrm{safe},i}} \mathcal{H}_i(z_i, \nabla_{z_i} V_i^*, u_{\mathrm{safe},i}) = 0
	\label{eq:hjb_safe}
\end{equation}
where the Hamiltonian $\mathcal{H}_i$ is defined as
\begin{equation}
    \begin{split}
        \mathcal{H}_i &\triangleq B_i(z_i) + u_{\mathrm{safe},i}^T R_i \, u_{\mathrm{safe},i} \\
        &\quad + (\nabla_{z_i} V_i^*)^T \Big( F_i(x) + G_i(x) u_{\mathrm{safe},i} \Big) - \alpha V_i^*
    \end{split}
    \label{eq:hamiltonian}
\end{equation}

Because the weighting matrix $R_i$ is strictly positive-definite, the Hamiltonian is globally strictly convex with respect to the control input $u_{\mathrm{safe},i}$. Therefore, the necessary and sufficient condition for optimality is the stationarity condition, $\frac{\partial \mathcal{H}_i}{\partial u_{\mathrm{safe},i}} = 0$ \cite[Theorem 5.8]{rudin1976principles}. Differentiating the Hamiltonian with respect to the control input yields
\begin{equation}
	\frac{\partial \mathcal{H}_i}{\partial u_{\mathrm{safe},i}} = 2 R_i \, u_{\mathrm{safe},i} + G_i(x)^T \nabla_{z_i} V_i^* = 0
\end{equation}

Solving for the optimal safety control law $u_{\mathrm{safe},i}^*$, we obtain
\begin{equation}
	u_{\mathrm{safe},i}^* = -\frac{1}{2} R_i^{-1} G_i(x)^T \nabla_{z_i} V_i^*
\end{equation}

Finally, substituting the explicit block structure of the input coupling matrix $G_i(x)$ from \eqref{eq:input_coupling} isolates the velocity components of the value function gradient. Let the full gradient be partitioned as $\nabla_{z_i} V_i^* = \big[ (\nabla_{z_{p,i}} V_i^*)^T, (\nabla_{z_{v,i}} V_i^*)^T \big]^T$. The matrix multiplication simplifies directly to $G_i(x)^T \nabla_{z_i} V_i^* = w_{\Sigma,i} \nabla_{z_{v,i}} V_i^*$. Thus, the closed-form optimal safety policy is
\begin{equation}
	u_{\mathrm{safe},i}^* = -\frac{1}{2} R_i^{-1} \, w_{\Sigma,i} \, \nabla_{z_{v,i}} V_i^*
	\label{eq:optimal_safe_control}
\end{equation}
where $z_{v,i} = z_i[d{+}1{:}2d]$ denotes exclusively the velocity components of the augmented danger state.

\subsection{Actor-Critic Neural Network}

\subsubsection{Critic Network}

To solve the Hamilton-Jacobi-Bellman (HJB) equation online, an adaptive critic neural network is constructed to approximate the unknown optimal value function $V_i^*(z_i)$. Assuming $V_i^*(z_i)$ is smooth and continuous over a compact operating domain, it can be represented by a linearly parameterized neural network according to the universal approximation theorem \cite{hornik1990universal}
\begin{equation}
	V_i^*(z_i) = W_{c,i}^{*T} \sigma_{c,i}(z_i) + \varepsilon_{c,i}(z_i)
\end{equation}
where $W_{c,i}^* \in \mathbb{R}^{N_c}$ represents the ideal, unknown critic weights, $\sigma_{c,i}(z_i) \in \mathbb{R}^{N_c}$ is the vector of activation functions, and $\varepsilon_{c,i}(z_i)$ is the bounded uniform approximation error.

In practice, the ideal weights are unknown, so the agent relies on the real-time estimate $\hat{W}_{c,i}(t)$. The estimated value function is given by
\begin{equation}
	\hat{V}_i(z_i) = \hat{W}_{c,i}^T \sigma_{c,i}(z_i)
	\label{eq:critic_approximation}
\end{equation}

By substituting the estimated value function \eqref{eq:critic_approximation} and the current actor-generated policy $\hat{u}_{\mathrm{safe},i}$ into the exact HJB equation \eqref{eq:hjb_safe}, the equation will no longer equal zero. Instead, it yields a residual discrepancy known as the Bellman error, $e_{c,i}(t) \in \mathbb{R}$
\begin{equation}
    \begin{split}
        e_{c,i} &= B_i(z_i) + \hat{u}_{\mathrm{safe},i}^T R_i \, \hat{u}_{\mathrm{safe},i} \\
        &\quad + \hat{W}_{c,i}^T \nabla_{z_i} \sigma_{c,i}(z_i) \dot{z}_i - \alpha \hat{W}_{c,i}^T \sigma_{c,i}(z_i)
    \end{split}
\end{equation}

To streamline the adaptive learning laws, we define the critic regressor vector $\omega_i(t) \in \mathbb{R}^{N_c}$ as
\begin{equation}
	\omega_i \triangleq \nabla_{z_i} \sigma_{c,i}(z_i) \dot{z}_i - \alpha \sigma_{c,i}(z_i)
	\label{eq:omega_regressor}
\end{equation}
So the Bellman error to be rewritten in the following compact, computationally tractable form
\begin{equation}
	e_{c,i} = B_i(z_i) + \hat{u}_{\mathrm{safe},i}^T R_i \, \hat{u}_{\mathrm{safe},i} + \hat{W}_{c,i}^T \omega_i
	\label{eq:bellman_error}
\end{equation}

To drive the estimated value function toward the optimal solution, the critic network weights are updated to minimize the instantaneous squared Bellman error, $E_i = \frac{1}{2} e_{c,i}^2$. We employ a normalized gradient descent algorithm, leading to the following continuous-time update law
\begin{equation}
	\dot{\hat{W}}_{c,i} = - \eta_c \frac{\omega_i}{(1 + \omega_i^T \omega_i)^2} \, e_{c,i}
	\label{eq:critic_update}
\end{equation}
where $\eta_c > 0$ dictates the critic learning rate. The squared normalization term $(1 + \omega_i^T \omega_i)^2$ in the denominator is a critical design feature; it strictly bounds the adaptation rate during sudden dynamic transients, preventing numerical instability when the regressor vector $\omega_i$ experiences extreme magnitudes during aggressive evasion maneuvers.

\subsubsection{Actor Network}

While the critic network estimates the optimal value function, the actor network is responsible for approximating the optimal safety policy $u_{\mathrm{safe},i}^*$ derived in \eqref{eq:optimal_safe_control}. The estimated safety policy, denoted as $\hat{u}_{\mathrm{safe},i} \in \mathbb{R}^d$, is formulated using a linearly parameterized neural network
\begin{equation}
	\hat{u}_{\mathrm{safe},i}(z_i) = \hat{W}_{a,i}^T \sigma_{a,i}(z_i)
\end{equation}
where $\hat{W}_{a,i} \in \mathbb{R}^{2d \times d}$ is the estimated actor weight matrix and $\sigma_{a,i}(z_i) \in \mathbb{R}^{2d}$ is the actor activation function vector. 

To train the actor network, we define an ideal target policy based on the current estimate of the value function provided by the critic. Substituting the critic's value function approximation $\hat{V}_i(z_i) = \hat{W}_{c,i}^T \sigma_{c,i}(z_i)$ into the analytical optimal control law yields the actor target
\begin{equation}
	u_{\mathrm{target},i} = -\frac{1}{2} R_i^{-1} w_{\Sigma,i} \nabla_{z_{v,i}} \sigma_{c,i}(z_i)^T \hat{W}_{c,i}
\end{equation}

The actor network seeks to minimize the discrepancy between its current output $\hat{u}_{\mathrm{safe},i}$ and this critic-derived target. We define the actor approximation error $e_{a,i} \in \mathbb{R}^d$ as
\begin{equation}
    \begin{split}
        e_{a,i} &\triangleq \hat{u}_{\mathrm{safe},i} - u_{\mathrm{target},i} \\
        &= \hat{W}_{a,i}^T \sigma_{a,i}(z_i) + \frac{1}{2} R_i^{-1} w_{\Sigma,i} \nabla_{z_{v,i}} \sigma_{c,i}(z_i)^T \hat{W}_{c,i}
    \end{split}
    \label{eq:actor_error}
\end{equation}

To continuously minimize the squared actor error $E_{a,i} = \frac{1}{2} e_{a,i}^T e_{a,i}$, the actor weights are updated using a continuous-time normalized gradient descent algorithm
\begin{equation}
	\dot{\hat{W}}_{a,i} = - \eta_a \frac{\sigma_{a,i}(z_i)}{1 + \sigma_{a,i}(z_i)^T \sigma_{a,i}(z_i)} \, e_{a,i}^T
	\label{eq:actor_update}
\end{equation}
where $\eta_a > 0$ is the strictly positive actor learning rate. Similar to the critic update, the normalization term in the denominator bounds the adaptation rate, ensuring that the gradient does not grow unboundedly when the danger state magnitude $\|z_i\|$ is large.

\begin{remark}
	The actor network requires an initially stabilizing (admissible) control policy to guarantee convergence. To satisfy this requirement, the actor weight matrix is warm-started rather than initialized to zero.
\end{remark}

To satisfy the Persistent Excitation (PE) condition a probing noise $n_i(t) \in \mathbb{R}^d$ is injected into the control channel. The total applied safety-modulated control input for follower $i$ is consequently augmented to
\begin{equation}
	u_i(t) = \rho_i \, u_{\mathrm{nom},i} + \hat{u}_{\mathrm{safe},i} + n_i(t)
	\label{eq:control_with_noise}
\end{equation}

\begin{remark}
	While continuous noise injection theoretically guarantees weight convergence, subjecting the multi-agent system to a constant-amplitude disturbance would permanently degrade the formation tracking accuracy and introduce steady-state chattering. To mathematically resolve this exploration-exploitation trade-off, the probing signal $n_i(t)$ is designed as a multi-frequency exploratory signal bounded by an exponentially decaying envelope $\|n_i(t)\| \le n_0 e^{-\kappa t}$.
	\label{rem:probing_noise}
\end{remark}

To synthesize the theoretical components of the actor-critic neural network, the danger-weighted state mapping, and the barrier modulation into a practical execution sequence, the complete algorithmic flow of the ADP safety controller is detailed in Algorithm \ref{alg:adp_safety}. 

\begin{algorithm}[htbp]
    \caption{Actor-Critic ADP Safety Controller for Follower $i$}
    \label{alg:adp_safety}
    \begin{algorithmic}[1] 
        \REQUIRE Current state $x_i(t) = (p_i, v_i)$, nominal control $u_{\mathrm{nom},i}$, safety distance $D_s$, sensor radius $R_{\rm sense}$, network weights $\hat{W}_{c,i}, \hat{W}_{a,i}$, and tuning parameters $\eta_c, \eta_a, \alpha, R_i, \beta, \mu$.
        \ENSURE Total safety-modulated control input $u_i(t)$ and weight derivatives $\dot{\hat{W}}_{c,i}, \dot{\hat{W}}_{a,i}$.
        
        \STATE Obtain local sensor measurements: $\mathcal{S}_i(t) = \{j \in \mathcal{V} \setminus \{i\} : \|p_i - p_j\| < R_{\rm sense}\}$
        \IF{$\mathcal{S}_i(t)$ is empty}
            \STATE $\rho_i \leftarrow 1$, $u_{\mathrm{safe},i} \leftarrow \mathbf{0}$, $u_i(t) \leftarrow u_{\mathrm{nom},i}$
            \STATE $\dot{\hat{W}}_{c,i} \leftarrow \mathbf{0}$, $\dot{\hat{W}}_{a,i} \leftarrow \mathbf{0}$ \quad \COMMENT{Networks pause when strictly safe}
        \ELSE
            \STATE \textbf{Phase 1: State Mapping \& Forward Pass}
            \STATE Initialize $h_{0,\min} \leftarrow \infty$, $z_i \leftarrow \mathbf{0} \in \mathbb{R}^{2d}$, $w_{\Sigma,i} \leftarrow 0$, $B_i \leftarrow 0$
            \FOR{\textbf{each} detected agent $j \in \mathcal{S}_i(t)$}
                \STATE $h_0(x_i, x_j) \leftarrow \|p_i - p_j\|^2 - D_s^2$
                \STATE $h_{0,\min} \leftarrow \min\big(h_{0,\min}, \; h_0(x_i, x_j)\big)$
                \STATE $w_{ij} \leftarrow D_s^4 / \big(h_0(x_i, x_j) + 2D_s^2\big)^2$
                \STATE $z_i \leftarrow z_i + w_{ij} [ (p_i - p_j)^T, (v_i - v_j)^T ]^T$
                \STATE $w_{\Sigma,i} \leftarrow w_{\Sigma,i} + w_{ij}$
                \STATE $B_i \leftarrow B_i + \text{BarrierPenalty}(h_{\mathrm{safe}}(x_i, x_j))$
            \ENDFOR
            \STATE $\rho_i \leftarrow 1 - \exp\Big(-\beta \max\big(0, h_{0,\min}\big)\Big)$
            
            \STATE \textbf{Phase 2: Policy Execution}
            \STATE $\sigma_{a,i}(z_i) \leftarrow z_i$
            \STATE $\hat{u}_{\mathrm{safe},i} \leftarrow \text{sat}(\hat{W}_{a,i}^T \sigma_{a,i}(z_i), \bar{U}_{\mathrm{NN},\max})$
            \STATE Calculate decaying probing noise $n_i(t)$
            \STATE Apply control: $u_i(t) \leftarrow \rho_i \, u_{\mathrm{nom},i} + \hat{u}_{\mathrm{safe},i} + n_i(t)$
            
            \STATE \textbf{Phase 3: Network Updates (Learning)}
            \STATE Calculate temporal derivative $\dot{z}_i$ based on relative kinematics.
            \STATE Compute basis vector $\sigma_{c,i}(z_i)$ and its Jacobian $\nabla_{z_i} \sigma_{c,i}(z_i)$
            \STATE Compute regressor: $\omega_i \leftarrow \nabla_{z_i} \sigma_{c,i}(z_i) \dot{z}_i - \alpha \sigma_{c,i}(z_i)$
            \STATE \textit{// Critic Network Update}
            \STATE $e_{c,i} \leftarrow B_i + \hat{u}_{\mathrm{safe},i}^T R_i \, \hat{u}_{\mathrm{safe},i} + \hat{W}_{c,i}^T \omega_i$
            \STATE $\dot{\hat{W}}_{c,i} \leftarrow - \eta_c \frac{\omega_i}{(1 + \omega_i^T \omega_i)^2} e_{c,i}$
            \STATE Apply soft weight projection to $\dot{\hat{W}}_{c,i}$ ensuring $\|\hat{W}_{c,i}\|_\infty \le \bar{W}_{c,\max}$
            
            \STATE \textit{// Actor Network Update}
            \STATE $e_{a,i} \leftarrow \hat{u}_{\mathrm{safe},i} + \frac{1}{2} R_i^{-1} w_{\Sigma,i} \nabla_{z_{v,i}} \sigma_{c,i}(z_i)^T \hat{W}_{c,i}$
            \STATE $\dot{\hat{W}}_{a,i} \leftarrow - \eta_a \frac{\sigma_{a,i}(z_i)}{1 + \sigma_{a,i}(z_i)^T \sigma_{a,i}(z_i)} e_{a,i}^T$
            \STATE Apply soft weight projection to $\dot{\hat{W}}_{a,i}$ ensuring $\|\hat{W}_{a,i}\|_\infty \le \bar{W}_{a,\max}$
        \ENDIF
        \RETURN $u_i(t)$
    \end{algorithmic}
\end{algorithm}

\subsection{Safety and Stability Analysis}

To formally establish the stability and safety of the proposed multi-agent control architecture, we first define the standard assumptions regarding the neural network approximation capabilities and the system's operating domain.

\begin{assumption}
	\label{as:nn_rigor}
	We assume that the optimal value function $V_i^*(z_i)$ is continuously differentiable on a compact operating domain $\mathcal{Z}$, satisfying $V_i^*(\mathbf{0}) = 0$ and $V_i^*(z_i) > 0$ for all $z_i \neq \mathbf{0}$. The actor and critic activation basis functions, along with their spatial gradients, are locally Lipschitz continuous and uniformly bounded. Furthermore, there exist ideal, unknown constant weights $W_{c,i}^*$ and $W_{a,i}^*$ bounded by strict positive constants $\bar{W}_c$ and $\bar{W}_a$, such that the network reconstruction errors and their spatial gradients are uniformly bounded by $\bar{\varepsilon}_c$, $\bar{\varepsilon}_a$, and $\bar{\varepsilon}_{c,\nabla}$. Finally, the probing noise injected into the control channel is bounded such that $\|n_i(t)\| \le \bar{N}$ for all $t \ge 0$.
\end{assumption}

\begin{assumption}
	\label{as:pe_condition}
	To guarantee the convergence of the actor-critic weight errors and prevent parameter drift, the probing noise $n_i(t)$ ensures that the generated regressor vectors are persistently exciting over the compact set $\mathcal{Z}$. Specifically, there exist strictly positive constants $\beta_1, \beta_2, \beta_3, \beta_4$ and a finite time window $T > 0$ such that for all $t \ge 0$, the integrated normalized regressors satisfy the two-sided bounds:
	\begin{equation}
		\beta_1 I \le \int_t^{t+T} \frac{\omega_i(\tau) \omega_i^T(\tau)}{(1 + \omega_i^T(\tau) \omega_i(\tau))^2} d\tau \le \beta_2 I
	\end{equation}
	\begin{equation}
		\beta_3 I \le \int_t^{t+T} \frac{\sigma_{a,i}(z_i(\tau)) \sigma_{a,i}^T(z_i(\tau))}{1 + \sigma_{a,i}^T(z_i(\tau)) \sigma_{a,i}(z_i(\tau))} d\tau \le \beta_4 I
	\end{equation}
\end{assumption}

Based on these assumptions, the theoretical analysis is divided into four primary theorems: absolute collision avoidance, neural network weight convergence, composite closed-loop stability, and nominal recovery.

\begin{theorem}[ADP Controller Safety]
	\label{thm:adp_safety}
	Consider the multi-agent system governed by the nonlinear dynamics \eqref{eq:dynamics} operating under the barrier-modulated actor-critic control law $u_i = \rho_i \, u_{\mathrm{nom},i} + \hat{u}_{\mathrm{safe},i} + n_i(t)$, where the actor network output is explicitly saturated such that $\|\hat{u}_{\mathrm{safe},i}\| \le \bar{U}_{\mathrm{NN},\max}$. Suppose Assumptions \ref{as:preliminary}, \ref{as:bounded_dynamics}, and \ref{as:nn_rigor} hold, and the initial states strictly reside within the safe set, such that $h_0(x_i(0), x_j(0)) > 0$ and $h_{\mathrm{safe}}(x_i(0), x_j(0)) > 0$ for all valid pairs $(i,j)$. 
	
	If the actor network's saturation limit $\bar{U}_{\mathrm{NN},\max}$ satisfies the sufficient boundary condition:
	\begin{equation}
		\bar{U}_{\mathrm{NN},\max} > \bar{U}_{\mathrm{nom}} + \bar{\Theta} + \bar{N} + \bar{U}_{\mathrm{other}} + \frac{\gamma^2(R_{\rm sense}^2 - D_s^2)}{4 D_s},
		\label{eq:adp_sufficient_bound}
	\end{equation}
	then the safe set $\mathcal{C} = \{x \mid h_0(x_i, x_j) > 0, \; \forall j \in \mathcal{S}_i\}$ is forward invariant, guaranteeing absolute collision avoidance for all $t \ge 0$.
\end{theorem}

\begin{proof}
	By Nagumo's Theorem \cite{nagumo1942lage}, $\mathcal{C}$ is forward invariant if $\dot{h}_{\mathrm{safe}}|_{h_{\mathrm{safe}} \to 0^+} > 0$. Let $p_{ij} \triangleq p_i - p_j$, $v_{ij} \triangleq v_i - v_j$. Applying HOCBF boundary conditions \cite{xiao2021high}, $2p_{ij}^T v_{ij} \to -\gamma h_0(x_i, x_j) \implies$
	\begin{equation}
		\dot{h}_{\mathrm{safe}}\Big|_{h_{\mathrm{safe}} \to 0^+} = 2\|v_{ij}\|^2 + 2p_{ij}^T \dot{v}_{ij} - \gamma^2 h_0(x_i, x_j).
		\label{eq:adp_h_dot_boundary}
	\end{equation}
	
	Substituting $\dot{v}_{ij} = (\hat{u}_{\mathrm{safe},i} - \hat{u}_{\mathrm{safe},j}) + \Delta_{ij}$, where $\Delta_{ij}$ consolidates all non-safety forces acting on the pair. Following the identical worst-case disturbance bounding procedure established in Theorem 2
	\begin{equation*}
    \begin{split}
        2p_{ij}^T \Delta_{ij} &\ge -4R_{\rm sense} \Big( \bar{U}_{\mathrm{nom}} + \bar{\Theta} + \bar{N} + \bar{U}_{\mathrm{other}} \Big), \\
        -\gamma^2 h_0(x_i, x_j) &\ge -\gamma^2(R_{\rm sense}^2 - D_s^2).
    \end{split}
\end{equation*}
	
	At the safety boundary ($h_{\mathrm{safe}} \to 0^+$), $B_i(z_i) \to \infty$ and $\|\nabla_{z_i} V_i^*\| \to \infty$. The optimal policy reduces strictly to the unconstrained linear inner product limit. By Lie derivative properties of control-affine systems \cite{khalil2002nonlinear}
	\begin{equation*}
    \begin{split}
        (\nabla_{z_i} V_i^*)^T G_i(x) &= -\nu(z_i) \left( \frac{\partial \dot{h}_{\mathrm{safe}}}{\partial u_{\mathrm{safe},i}} \right)^T \\
        &= -2\nu(z_i) p_{ij}^T, \quad \nu(z_i) > 0.
    \end{split}
\end{equation*}
	
	$\implies \hat{u}_{\mathrm{safe},i} = \arg\max_{\|u\| \le \bar{U}_{\mathrm{NN},\max}} (2p_{ij}^T u)$. By standard convex optimization principles for linear projections over a Euclidean ball \cite{boyd2004convex}, the optimal vectors lie strictly on the boundary
	\begin{equation*}
		\hat{u}_{\mathrm{safe},i} = \bar{U}_{\mathrm{NN},\max} \frac{p_{ij}}{\|p_{ij}\|}, \quad \hat{u}_{\mathrm{safe},j} = -\bar{U}_{\mathrm{NN},\max} \frac{p_{ij}}{\|p_{ij}\|}.
	\end{equation*}
	
	$\implies 2p_{ij}^T \big(\hat{u}_{\mathrm{safe},i} - \hat{u}_{\mathrm{safe},j}\big) = 4 \|p_{ij}\| \bar{U}_{\mathrm{NN},\max}$. Substituting bounds into the boundary derivative \eqref{eq:adp_h_dot_boundary} and noting $2\|v_{ij}\|^2 \ge 0$
    \begin{equation*}
    \begin{split}
        \dot{h}_{\mathrm{safe}}\Big|_{h_{\mathrm{safe}} \to 0^+} 
        &\ge 4\|p_{ij}\| \bar{U}_{\mathrm{NN},\max} \\
        &\quad - 4R_{\rm sense} \Big( \bar{U}_{\mathrm{nom}} + \bar{\Theta} + \bar{N} + \bar{U}_{\mathrm{other}} \Big) \\
        &\quad - \gamma^2(R_{\rm sense}^2 - D_s^2).
    \end{split}
    \end{equation*}
	
	Since physical safety requires $\|p_{ij}\| \ge D_s$, satisfying condition \eqref{eq:adp_sufficient_bound} guarantees $\dot{h}_{\mathrm{safe}}|_{h_{\mathrm{safe}} \to 0^+} > 0 \implies \mathcal{C}$ is forward invariant.
\end{proof}

\begin{theorem}[ADP Controller Stability]
	\label{thm:composite_uub}
	Consider the multi-agent system governed by the nonlinear dynamics with the barrier-modulated actor-critic control law. Suppose the assumptions regarding the nominal formation tracking, neural network approximations, and the Persistence of Excitation (PE) conditions hold. Let the initial states strictly reside within the safe set. Then, the composite closed-loop system comprising the nominal affine tracking and adaptive estimation errors, the critic network weight errors $\tilde{W}_{c,i}$, and the actor network weight errors $\tilde{W}_{a,i}$ is Uniformly Ultimately Bounded (UUB).
\end{theorem}

\begin{proof}
	Let $\tilde{W}_{c,i} \triangleq W_{c,i}^* - \hat{W}_{c,i}$ and $\tilde{W}_{a,i} \triangleq W_{a,i}^* - \hat{W}_{a,i}$. Since $W_{c,i}^*, W_{a,i}^*$ are constant, $\dot{\tilde{W}}_{c,i} = -\dot{\hat{W}}_{c,i}$ and $\dot{\tilde{W}}_{a,i} = -\dot{\hat{W}}_{a,i}$. Consider the composite Lyapunov candidate
	\begin{equation*}
		V_{\mathrm{total}} = V_{\mathrm{nom}}(Z_{\mathrm{track}}) + \sum_{i \in \mathcal{V}_f} \left( \frac{1}{2} \tilde{W}_{c,i}^T \tilde{W}_{c,i} + \frac{1}{2} \operatorname{tr}(\tilde{W}_{a,i}^T \tilde{W}_{a,i}) \right).
	\end{equation*}
	
	The total control $u_i = u_{\mathrm{nom},i} + \varepsilon_{\mathrm{nom},i}$, where $\varepsilon_{\mathrm{nom},i} \triangleq (\rho_i - 1) u_{\mathrm{nom},i} + \hat{u}_{\mathrm{safe},i} + n_i(t)$. Boundedness of its components ensures $\|\varepsilon_{\mathrm{nom}}\| \le \bar{\varepsilon}_{\mathrm{nom}}$. Evaluating $\dot{V}_{\mathrm{nom}}$ yields $\dot{V}_{\mathrm{nom}} \le -\alpha\|Z_{\mathrm{track}}\|^2 + Z_{\mathrm{track}}^T \varepsilon_{\mathrm{nom}}$. Applying standard robust control bounding for non-vanishing perturbations \cite{khalil2002nonlinear}
	\begin{equation*}
		\dot{V}_{\mathrm{nom}} \le -\frac{\alpha}{2}\|Z_{\mathrm{track}}\|^2 + \frac{\bar{\varepsilon}_{\mathrm{nom}}^2}{2\alpha}.
	\end{equation*}
	
	Substituting the normalized gradient descent update laws for the actor and critic networks, and applying Young's inequality to the residual approximation errors \cite{vamvoudakis2010online}, the weight error derivatives satisfy
	\begin{align*}
		\tilde{W}_{c,i}^T \dot{\tilde{W}}_{c,i} &\le -\frac{\eta_c}{2} \tilde{W}_{c,i}^T \Omega_{c,i} \tilde{W}_{c,i} + \frac{\eta_c}{2} \bar{\varepsilon}_{H,i}^2, \\
		\operatorname{tr}(\tilde{W}_{a,i}^T \dot{\tilde{W}}_{a,i}) &\le -\frac{\eta_a}{2} \operatorname{tr}(\tilde{W}_{a,i}^T \Omega_{a,i} \tilde{W}_{a,i}) + \frac{\eta_a}{2} \bar{\varepsilon}_{a,i}^2,
	\end{align*}
	where $\Omega_{c,i} \triangleq \frac{\omega_i \omega_i^T}{(1 + \omega_i^T \omega_i)^2}$ and $\Omega_{a,i} \triangleq \frac{\sigma_{a,i} \sigma_{a,i}^T}{1 + \sigma_{a,i}^T \sigma_{a,i}}$ are the normalized regressor matrices.
	
	Combining these bounds into $\dot{V}_{\mathrm{total}}$ yields
	\begin{equation*}
    \begin{split}
        \dot{V}_{\mathrm{total}} &\le -\frac{\alpha}{2} \|Z_{\mathrm{track}}\|^2 \\
        &\quad - \sum_{i \in \mathcal{V}_f} \left( \frac{\eta_c}{2} \tilde{W}_{c,i}^T \Omega_{c,i} \tilde{W}_{c,i} + \frac{\eta_a}{2} \operatorname{tr}(\tilde{W}_{a,i}^T \Omega_{a,i} \tilde{W}_{a,i}) \right) \\
        &\quad + C_{\mathrm{total}},
    \end{split}
    \end{equation*}
	where $C_{\mathrm{total}} \triangleq \frac{\bar{\varepsilon}_{\mathrm{nom}}^2}{2\alpha} + \sum_{i \in \mathcal{V}_f} \left( \frac{\eta_c}{2} \bar{\varepsilon}_{H,i}^2 + \frac{\eta_a}{2} \bar{\varepsilon}_{a,i}^2 \right)$.
	
	While $\Omega_{c,i}$ and $\Omega_{a,i}$ are instantaneous rank-1 matrices, Assumption \ref{as:pe_condition} guarantees Persistent Excitation. By the Lyapunov extension theorem for persistently excited systems \cite{khalil2002nonlinear}, integrating $\dot{V}_{\mathrm{total}}$ over the window $[t, t+T]$ yields strictly positive effective decay rates $\beta_c \triangleq \frac{\eta_c \beta_1}{2T}$ and $\beta_a \triangleq \frac{\eta_a \beta_3}{2T}$
	\begin{equation*}
    \begin{split}
        \dot{\bar{V}}_{\mathrm{total}} &\le -\frac{\alpha}{2} \|Z_{\mathrm{track}}\|^2 \\
        &\quad - \sum_{i \in \mathcal{V}_f} \left( \beta_c \|\tilde{W}_{c,i}\|^2 + \beta_a \|\tilde{W}_{a,i}\|_F^2 \right) + C_{\mathrm{total}}.
    \end{split}
    \end{equation*}
	
	Consequently, $\dot{\bar{V}}_{\mathrm{total}} < 0$ strictly holds whenever:
	\begin{equation*}
    \begin{split}
        \|Z_{\mathrm{track}}\| &> \sqrt{\frac{2 C_{\mathrm{total}}}{\alpha}}, \\
        \|\tilde{W}_{c,i}\| &> \sqrt{\frac{C_{\mathrm{total}}}{\beta_c}}, \\
        \text{or} \quad \|\tilde{W}_{a,i}\|_F &> \sqrt{\frac{C_{\mathrm{total}}}{\beta_a}}.
    \end{split}
    \end{equation*}
	
	$\implies$ The composite system state $Z_{\mathrm{track}}$, $\tilde{W}_{c,i}$, and $\tilde{W}_{a,i}$ is Uniformly Ultimately Bounded (UUB).
\end{proof}

\section{Simulation}

We consider the practical scenario where UAVs are subject to air resistance disturbances \cite{qi2019modeling}. When a UAV is flying at high speeds, the air resistance is proportional to the square of its velocity, which can be modeled by $K v_i^2$. Therefore, we define the regressor matrix $\Phi_i$ as a diagonal matrix whose diagonal elements are given by $-v_{i,h}^2$, for $h \in \{1, 2, 3\}$. Consequently, $\theta_i$ encompasses the unknown parameter information, such as the air density, drag coefficient, and other aerodynamic factors.

We simulate a multi-agent system operating in a three-dimensional space consisting $9$ agents. The agent group is partitioned into $4$ leaders and $5$ followers. The nominal formation is shown in Fig.\ref{fig:dart_formation} and the stress matrix $\omega$ is given in \eqref{eq:stress_matrix}. The critic activation vector $\sigma_{c,i}(z_i)$ is specifically structured to consist of all unique quadratic monomials of the danger state $z_i$. The elements of the regressor vector are saturated such that $\|\omega_i\|_{\infty} \le \bar{\omega}_{\max}$, and the network weights are clamped via a projection operator to a known compact domain $\Omega_c = \{ \hat{W}_{c,i} \mid \|\hat{W}_{c,i}\|_{\infty} \le \bar{W}_{c,\max} \}$. The actor activation function is selected as the identity mapping, such that $\sigma_{a,i}(z_i) \equiv z_i$. The initial weights corresponding to the position component of the danger state ($z_{p,i}$) are initialized as $\hat{W}_{a,i}(0) = [k_{\mathrm{init}} I_d, \; \mathbf{0}_{d \times d}]^T$, where $k_{\mathrm{init}} > 0$. To maintain physical feasibility and prevent numerical instability during severe conflict scenarios, the output of the actor network is explicitly saturated such that $\|\hat{u}_{\mathrm{safe},i}\|_{\infty} \le \bar{U}_{\mathrm{NN},\max}$. Furthermore, a soft weight projection algorithm is applied to the adaptation law \eqref{eq:actor_update}, clamping the actor weights to the valid parameter space $\mathcal{W}_a = \{ \hat{W}_{a,i} \mid \|\hat{W}_{a,i}\|_{\infty} \le \bar{W}_{a,\max} \}$. If an element of $\hat{W}_{a,i}$ reaches the boundary and the derivative $\dot{\hat{W}}_{a,i}$ points outward, that element's update rate is strictly set to zero.

For demonstrating the performance of the proposed controllers an extreme unsafe situation is evaluated, where the leaders intentionally collapse to a single point. The 3D spatial trajectories for the barrier-gradient and ADP controllers are illustrated in Fig. \ref{fig:3D_gradient} and Fig. \ref{fig:3D_adp}, respectively. As shown in Fig. \ref{fig:min_distance}, the safety controllers actively maintains the minimum allowable distance between the follower-related pairs. Fig.\ref{fig:tracking_error} and Fig.\ref{fig:control_inputs} show the tracking error and controls input through the maneuvers. Finally, Fig. \ref{fig:critic_weights} and Fig. \ref{fig:actor_weights} shows the convergence of neural network weights.

\begin{table*}[t]
\normalsize
\begin{equation}
    \Omega = \left[
    \begin{array}{rrrrrrrrr}
        0.5928 & -0.0920 & -0.3519 & -0.0934 & -0.3519 &  0.1475 &  0.0000 &  0.1489 &  0.0000 \\
       -0.0920 &  0.1343 & -0.0423 &  0.0000 &  0.0460 & -0.1343 &  0.0883 &  0.0000 &  0.0000 \\
       -0.3519 & -0.0423 &  0.5234 &  0.0467 &  0.0000 &  0.0000 & -0.3052 & -0.0889 &  0.2182 \\
       -0.0934 &  0.0000 &  0.0467 &  0.1344 & -0.0410 &  0.0000 &  0.0000 & -0.1344 &  0.0877 \\
       -0.3519 &  0.0460 &  0.0000 & -0.0410 &  0.5228 & -0.0870 &  0.2170 &  0.0000 & -0.3059 \\
        0.1475 & -0.1343 &  0.0000 &  0.0000 & -0.0870 &  0.3657 & -0.3052 &  0.2315 & -0.2182 \\
        0.0000 &  0.0883 & -0.3052 &  0.0000 &  0.2170 & -0.3052 &  0.5222 & -0.2170 &  0.0000 \\
        0.1489 &  0.0000 & -0.0889 & -0.1344 &  0.0000 &  0.2315 & -0.2170 &  0.3658 & -0.3059 \\
        0.0000 &  0.0000 &  0.2182 &  0.0877 & -0.3059 & -0.2182 &  0.0000 & -0.3059 &  0.5241 
    \end{array}
    \right].
    \label{eq:stress_matrix}
\end{equation}
\vspace*{4pt}
\end{table*}

\begin{figure}[t]
    \centering
    
    \begin{minipage}{0.55\linewidth}
        \centering
        \includegraphics[trim=2cm 2cm 2cm 2cm, clip, width=\linewidth]{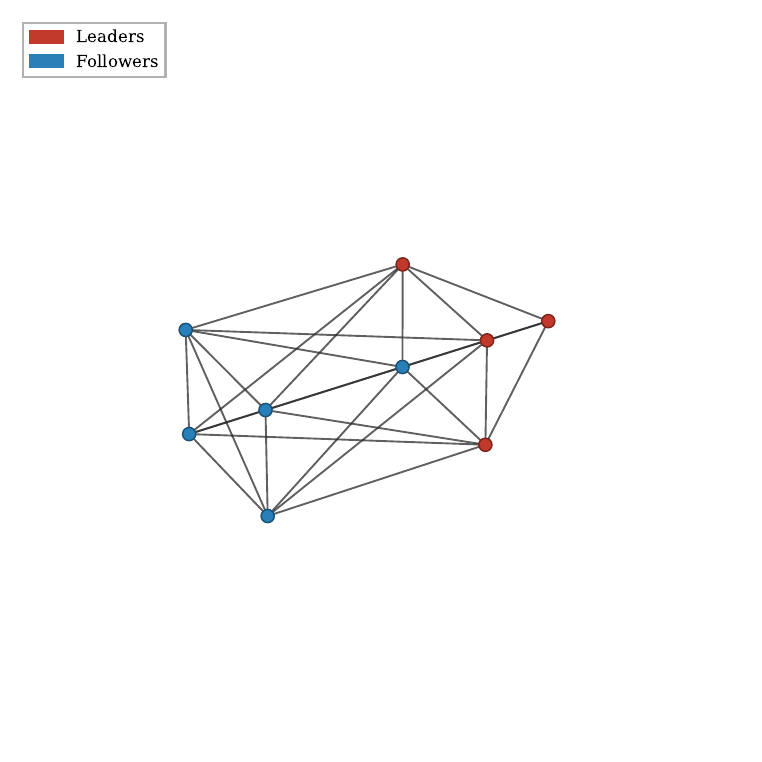} 
    \end{minipage}\hfill
    \begin{minipage}{0.4\linewidth}
        \begin{equation*}
            r = \begin{bsmallmatrix} 
                3.0 & 0.0 & 0.0 \\
                0.0 & 2.0 & 2.0 \\
                0.0 & -2.0 & 2.0 \\
                0.0 & -2.0 & -2.0 \\
                0.0 & 2.0 & -2.0 \\
                -6.0 & 2.0 & 2.0 \\
                -6.0 & -2.0 & 2.0 \\
                -6.0 & -2.0 & -2.0 \\
                -6.0 & 2.0 & -2.0
            \end{bsmallmatrix}      
            \label{eq:nominal_config}
        \end{equation*}
    \end{minipage}
    
    \vspace{1em}
    \caption{The nominal 3D affine dart formation consisting of 4 leaders (red) and 5 followers (blue). The black lines denote the 25 active edges. $r$ is the nominal configuration}
    \label{fig:dart_formation}
\end{figure}

\begin{figure}[t]
    \centering
    \includegraphics[width=\linewidth]{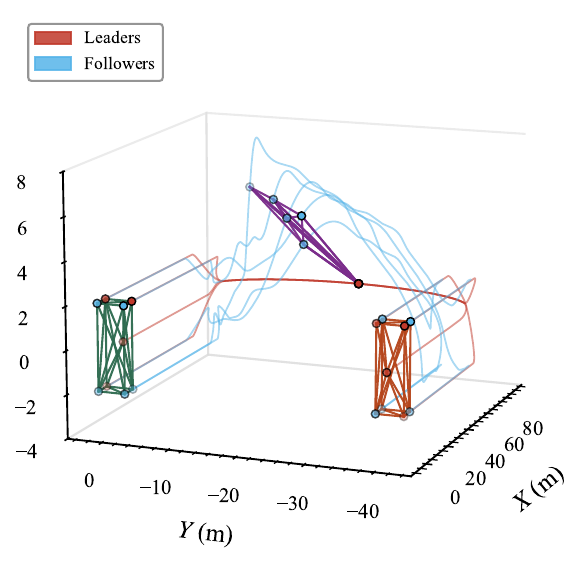}
    \caption{The 3D trajectory for the barrier-gradient based controller.}
    \label{fig:3D_gradient}
\end{figure}

\begin{figure}[t]
    \centering
    \includegraphics[width=\linewidth]{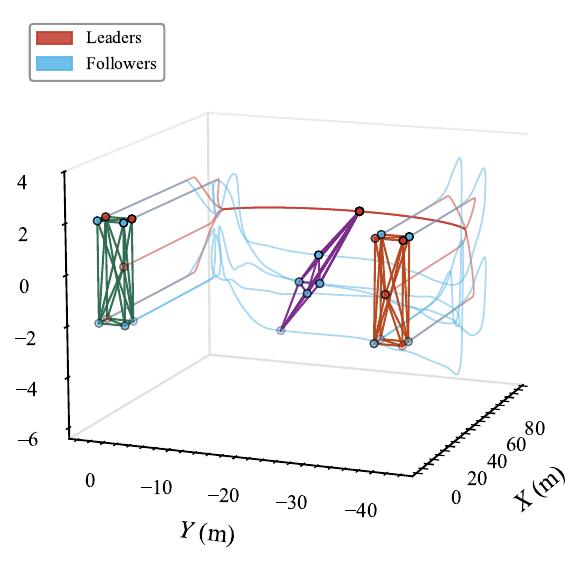}
    \caption{The 3D trajectory for the adp based controller.}
    \label{fig:3D_adp}
\end{figure}

\begin{figure}[t]
    \centering
    \includegraphics[width=\linewidth]{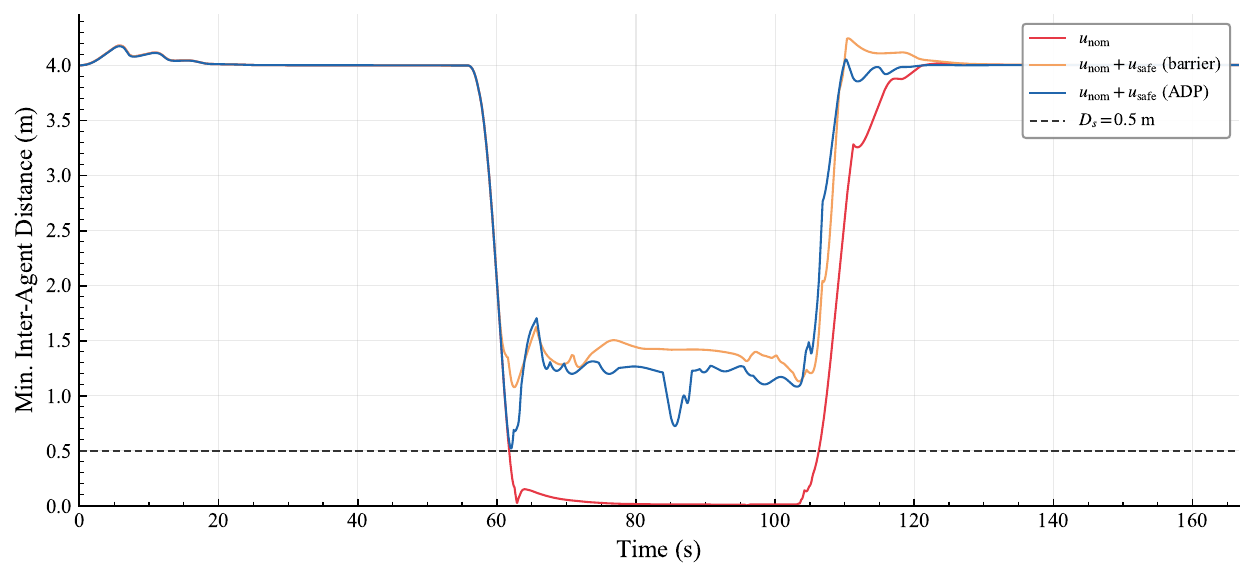}
    \caption{Comparison of minimum distance between follower related pairs.}
    \label{fig:min_distance}
\end{figure}

\begin{figure}[t]
    \centering
    \includegraphics[width=\linewidth]{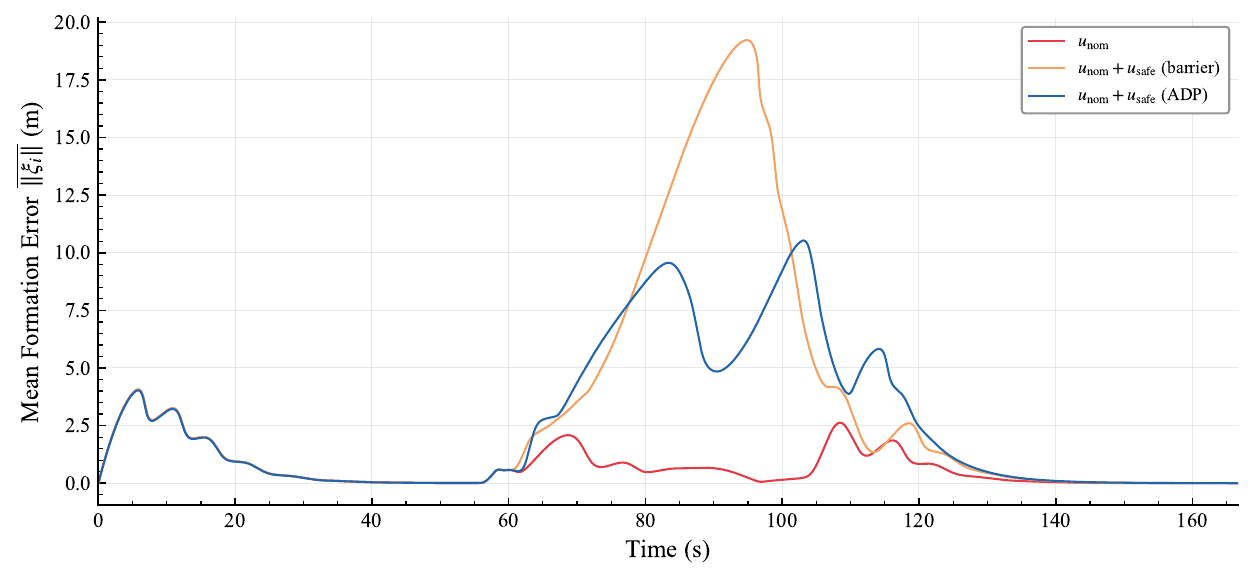}
    \caption{Comparison of tracking error.}
    \label{fig:tracking_error}
\end{figure}

\begin{figure}[t]
    \centering
    \includegraphics[width=\linewidth]{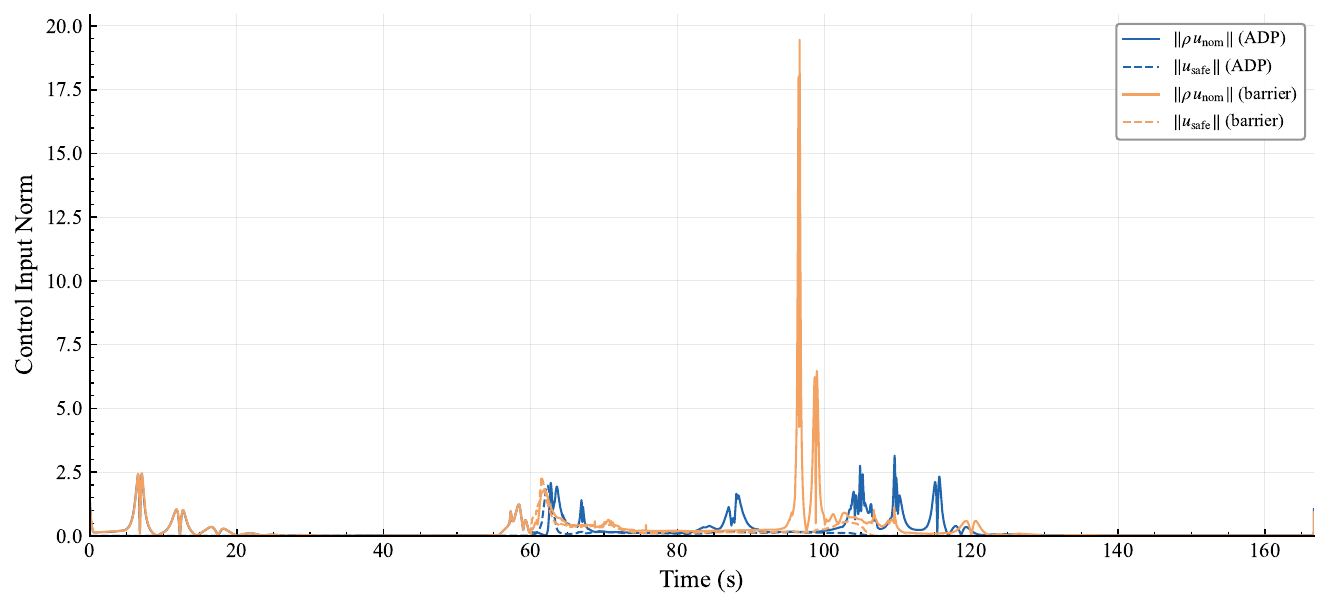}
    \caption{Comparison of control inputs.}
    \label{fig:control_inputs}
\end{figure}

\begin{figure}[t]
    \centering
    \includegraphics[width=\linewidth]{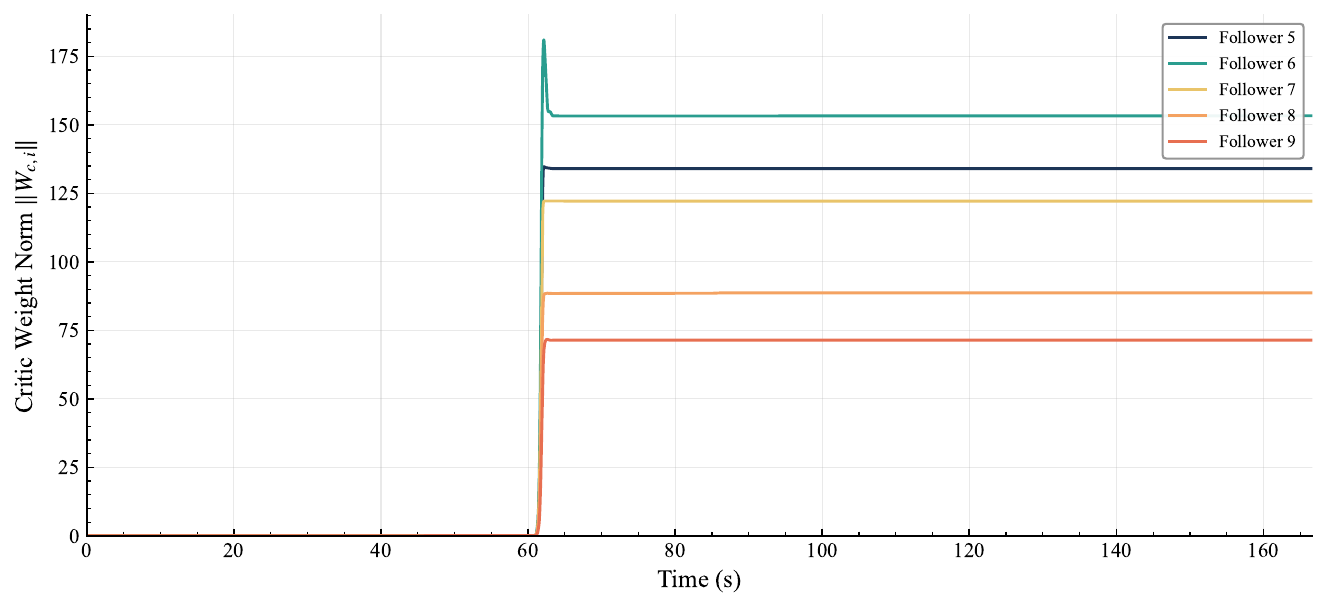}
    \caption{Critic weights convergence.}
    \label{fig:critic_weights}
\end{figure}

\begin{figure}[t]
    \centering
    \includegraphics[width=\linewidth]{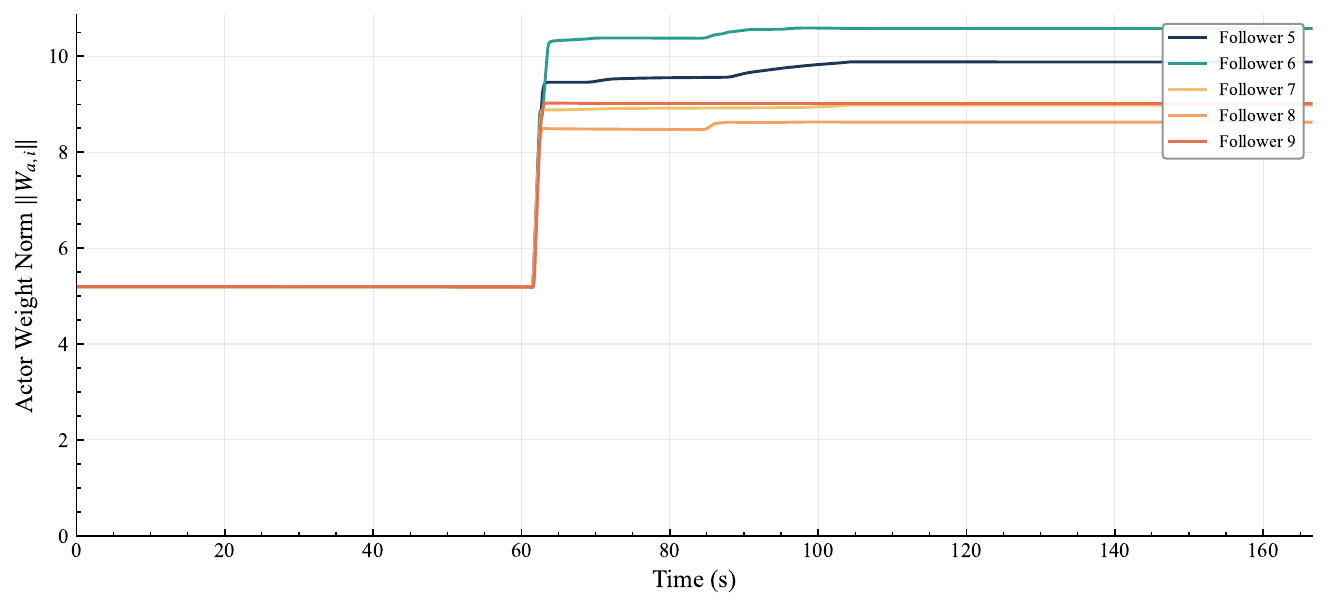}
    \caption{Actor weights convergence.}
    \label{fig:actor_weights}
\end{figure}

\section{Conclusion}
This paper investigated the problem of safe affine formation control for second-order multi-agent systems subject to parametric uncertainties. To achieve simultaneous formation tracking and collision avoidance, a novel barrier-modulated control architecture is proposed. By integrating Higher-Order Control Barrier Functions (HOCBFs), the framework smoothly prioritizes safety over nominal tracking near critical boundaries. Two distinct safety controllers are developed and analyzed, a computationally efficient analytical barrier-gradient controller and an optimal data-driven adaptive dynamic programming (ADP) controller. The ADP approach leveraged an actor-critic neural network to learn optimal safety policies online without requiring exact system dynamics. Rigorous theoretical analysis utilizing Nagumo's theorem and Lyapunov stability demonstrated that both proposed controllers guarantee the forward invariance of the safe set and ensure that all closed-loop signals, including tracking and estimation errors, remain Uniformly Ultimately Bounded (UUB). Finally, simulation results confirmed the validity, robustness, and performance of the proposed methods in maintaining safe affine formations during complex maneuvers. Future work will focus on extending this framework to directed communication topologies and the generalized nonlinear control affine system.

\clearpage
\bibliographystyle{IEEEtran}
\bibliography{references}

\end{document}